\documentclass[letterpaper]{article}

\usepackage[english]{babel}
\usepackage[utf8]{inputenc}
\usepackage[T1]{fontenc}

\usepackage{graphicx} 
\usepackage[margin=2.15cm]{geometry}
\usepackage{algorithm}
\usepackage[noend]{algpseudocode}
\usepackage{amsthm}
\usepackage{amssymb,amsmath}  
\usepackage{thmtools}  
\usepackage{mathtools}
\usepackage[colorlinks=true, allcolors=blue]{hyperref}
\usepackage[capitalize]{cleveref}
\usepackage{todonotes}
\usepackage{placeins}
\usepackage{subcaption}
\usepackage{booktabs}
\usepackage{microtype}

\usepackage{newtxtext}
\usepackage{newtxmath}

\usepackage{enumitem}
\setlist[1]{labelindent=\parindent}
\setlist[enumerate]{label=(\arabic*)}
\setlist[itemize]{noitemsep}

\usepackage{titling}
\usepackage{environ}
\NewEnviron{fix}{\textcolor{Fuchsia}{\BODY}\normalcolor}

\usepackage[dvipsnames]{xcolor}

\newcommand{\cO}{\mathcal{O}}

\newcommand{\setsize}[1]{\left|#1\right|}
\DeclareMathOperator{\degreeOp}{\delta}
\DeclareMathOperator{\indegreeOp}{\degreeOp^{-}}
\DeclareMathOperator{\outdegreeOp}{\degreeOp^{+}}
\newcommand{\indegree}[1]{\indegreeOp\hspace{-0.2em}\left(#1\right)}
\newcommand{\outdegree}[1]{\outdegreeOp\hspace{-0.2em}\left(#1\right)}

\newcommand{\N}{\mathbb{N}}

\newcommand{\totalETs}{z_T}
\newcommand{\totalArbs}{z_A}
\newcommand{\remainingGraph}[1]{G'_{#1}}
\newcommand{\prefix}[1]{P_{#1}}

\newcommand{\scc}[1]{\mathrm{SCC}(#1)}
\newcommand{\ie}[1]{(i.e.,~#1)}
\newcommand{\eg}[1]{(e.g.,~#1)}

\definecolor{CWIBlue}{rgb}{0.38, 0.478, 0.56}
\definecolor{CWIDGreen}{rgb}{0.372, 0.494, 0.419}
\definecolor{CWILGreen}{rgb}{0.603, 0.717, 0.486}
\definecolor{CWIRed}{rgb}{0.686, 0.188, 0.262}

\newif\ifhidetodos
\hidetodostrue

\usepackage{marginnote}

\newcommand{\cbox}[2][yellow]{%
  \fcolorbox{#1}{white}{\parbox{\dimexpr\linewidth-2\fboxsep}{\strut #2\strut}}%
}

\newcommand{\customtodo}[3]{\textcolor{#2}{\(\blacktriangledown\)}\marginnote{\raggedright \textcolor{#2}{\textbf{#1:} #3}}}
\newcommand{\customtododisplay}[3]{\noindent\cbox[#2]{\textcolor{#2}{\textbf{#1:} #3}}}
\newcommand{\customtodoinline}[3]{\textcolor{#2}{\textcolor{#2}{\textbf{#1:} #3}}}

\ifhidetodos
    \renewcommand{\customtodo}[3]{}
    \renewcommand{\customtododisplay}[3]{}
    \renewcommand{\customtodoinline}[3]{}
\fi

\usepackage{authblk}

\title{Optimal Enumeration of Eulerian Trails \\ in Directed Graphs}

\author[1,2]{Ben Bals}
\author[1,2,3]{Solon P. Pissis}
\author[2]{Matei Tinca}
\affil[1]{CWI, Amsterdam, The Netherlands}
\affil[2]{Vrije Universiteit, Amsterdam, The Netherlands}
\affil[3]{The Cyprus Institute, Nicosia, Cyprus}
\date{\today}

\declaretheorem[numberwithin=section]{theorem}

\declaretheorem[sibling=theorem]{corollary}
\declaretheorem[sibling=theorem]{lemma}

\declaretheorem[sibling=theorem]{fact}

\declaretheorem[sibling=theorem]{observation}

\declaretheorem[name=Compression Rule]{compression}
\Crefname{compression}{Compression Rule}{Compression Rules}
\declaretheorem[name=Multigraph Compression Rule]{mcompression}
\Crefname{mcompression}{Multigraph Compression Rule}{Multigraph Compression Rules}

\begin{document}

\maketitle

\begin{abstract}
The BEST theorem, due to de Bruijn, van Aardenne-Ehrenfest, Smith, and Tutte,
is a classical tool from graph theory
that links the Eulerian trails in a directed graph $G=(V,E)$ 
with the arborescences in $G$.
In particular, one can use the BEST theorem to count the Eulerian trails in $G$
in polynomial time.
For enumerating the Eulerian trails in $G$, one could naturally resort to first enumerating
the arborescences in $G$ and then exploiting the insight of the BEST theorem to enumerate the Eulerian trails in $G$: every arborescence in $G$ corresponds to at least one Eulerian trail in $G$.
For over two decades, the fastest algorithm for enumerating arborescences in $G$ took $\cO(m\log n + n + \totalArbs \log^2 n)$ time, where $n=|V|$, $m=|E|$, and $\totalArbs$ is the number of arborescences in $G$ [Uno, ISAAC 1998].
Since Uno's algorithm does not lead to an optimal enumeration of Eulerian trails in directed graphs, we were motivated to develop a direct algorithm for this problem.\footnote{Independently of and concurrently with our work, the running time to enumerate arborescences was improved to the optimal $\cO(m + \totalArbs)$ time~\cite{gawrychowski2026enumeratingdirectedspanningtrees} using a rather intricate technique. We believe a \emph{direct} and \emph{simple} algorithm to enumerate Eulerian trails in directed graphs remains of interest.} 

Our central contribution is a remarkably simple algorithm to \emph{directly} enumerate the $\totalETs$ Eulerian trails in $G$ in the \emph{optimal} $\cO(m + \totalETs)$ time. 
As a consequence, our result improves on an implementation of the BEST theorem for counting Eulerian trails in $G$ when $\totalETs=o(n^2)$, and also unconditionally improves the combinatorial $\cO(m\cdot \totalETs)$-time algorithm of Conte et al.~[TKDD 2026] for the same task. Moreover, we show that, with some care, our algorithm can be extended to enumerate Eulerian trails in directed multigraphs in optimal time, enabling applications in bioinformatics and data privacy.
\end{abstract}

\section{Introduction}\label{sec:intro}

Let $G=(V,E)$ be a graph on $n=|V|$ nodes and $m=|E|$ edges.
Unless explicitly stated otherwise, we will assume throughout that $G$ is simple and connected when $G$ is undirected, or simple and weakly connected when $G$ is directed.
A \emph{trail} in $G$ is a walk in $G$ with no repeated edge.
An \emph{Eulerian trail} is a trail that uses each graph edge exactly once. If $G$ has an Eulerian trail, we will call $G$ \emph{Eulerian}.
An \emph{Eulerian cycle}, or \emph{Eulerian circuit}, is an Eulerian trail that starts and ends at the same node in $G$; that is, it is a graph cycle that uses each edge exactly once.

Euler showed that an undirected graph has an Eulerian cycle if and only if it has no nodes of odd degree. Similarly, an undirected graph has an Eulerian trail if and only if exactly zero or two nodes have odd degree. A directed graph has an Eulerian cycle if and only if for every node $u\in V$ the in-degree $\indegree{u}$ is equal to the out-degree $\outdegree{u}$. Similarly, a directed graph has an Eulerian trail if and only if for at most one node $u\in V$, $\outdegree{u}-\indegree{u}=1$, and for at most one node $v$, $\indegree{v}-\outdegree{v}=1$, and every other node has equal in-degree and out-degree. 
Thus, \emph{deciding} whether an undirected or directed graph $G$ has an Eulerian trail can be done in $\cO(m)$ time. If $G$ is Eulerian, constructing a \emph{witness} Eulerian trail can be done in $\cO(m)$ time for both undirected and directed graphs using the algorithm by Hierholzer~\cite{hierholzer}. In fact, Hierholzer discovered this algorithm before his death in 1871, and it was published posthumously in 1873.

Eulerian trails are arguably among the most basic combinatorial objects.
Most undergraduate or graduate courses on graph algorithms provide the above introduction to Eulerian graphs.
Motivated by the crossing of the iconic Seven Bridges of Königsberg, this topic often appears in the first lecture.
As with many other combinatorial objects, the two most basic remaining questions are \emph{counting} and \emph{enumeration}. 

For undirected graphs, Brightwell and Winkler~\cite{DBLP:conf/alenex/BrightwellW05} showed that counting Eulerian trails is $\textsf{\#P}$-complete. Kurita and Wasa~\cite{DBLP:journals/tcs/KuritaW22} showed that enumerating all $\totalETs$ Eulerian trails in $G$ can be done in the optimal $\cO(m+\totalETs)$ time.

For directed graphs, counting Eulerian trails can be done in polynomial time using the BEST theorem~\cite{best}, named after de Bruijn, van Aardenne-Ehrenfest, Smith, and Tutte. The BEST theorem relies on the close relationship between Eulerian trails and arborescences.
A directed graph is called an \emph{arborescence} if, from a given node $s$ known as the root node, there is exactly one elementary path (meaning without node repetition) from $s$ to every other graph node.
Let us fix the source and target nodes of the Eulerian trails in graph $G$, and denote them by $s$ and $t$, respectively.
Further, let $r_u:= \outdegree{u}$ for $u\neq t$, and $r_t := \outdegree{t}+1$.
The BEST theorem can be stated as follows (see also \Cref{fig:arbo-trails}):
\begin{equation}\label{eq:BEST} 
  \totalETs = \totalArbs \cdot\left(\prod_{u\in V}(r_u-1)!\right),
\end{equation}
where $\totalETs$ is the number of Eulerian trails that start at node $s$,
and $\totalArbs$ is the number of arborescences rooted at node $s$. Thus, $\totalETs$ is obtained by multiplying $\totalArbs$ by the number of permutations of the edges outgoing from each node ($\prod_{u\in V}(r_u-1)!$).   
One can compute $\totalArbs$ as a determinant~\cite{DBLP:books/daglib/0009415} by Kirchhoff's theorem~\cite{Tutte1948}.
Since computing a determinant is equivalent to matrix multiplication~\cite{det}, counting Eulerian trails in $G$ requires $\cO(\textrm{poly}(m))$ time.
However, as matrix multiplication requires $\Omega(n^2)$ time and space, and since the lower-upper (LU) decomposition of the underlying Laplacian matrix is susceptible to numerical issues~\cite{DBLP:journals/siammax/Foster94,DBLP:books/daglib/0086372,DBLP:books/daglib/0007208}, Conte, Grossi, Loukides, Pisanti, Pissis, and Punzi~\cite{TKDD2025} developed an alternative, combinatorial algorithm for counting $\totalETs$ exactly in $\cO(m\cdot \totalETs)$ time; in fact, their goal was to \emph{assess} whether $\totalETs \geq z$, for some input parameter $z$.

\begin{figure}[t]
    \centering
    \includegraphics[width=\linewidth]{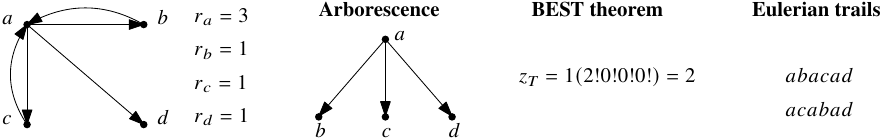}
    \caption{A directed graph, the arborescence rooted at node $a$, and the two Eulerian trails starting at $a$.}
    \label{fig:arbo-trails}
\end{figure}

For enumerating Eulerian trails in directed graphs, one could naturally resort to first enumerating the arborescences in $G$ and then exploiting the insight of the BEST theorem to enumerate the Eulerian trails in $G$: every arborescence corresponds to at least one Eulerian trail. For over two decades, the fastest algorithm for enumerating arborescences in $G$ took $\cO(m\log n + n + \totalArbs \log^2 n)$ time, where $n=|V|$, $m=|E|$, and $\totalArbs$ is the number of arborescences in $G$~\cite{DBLP:conf/isaac/Uno98}. Since Uno's algorithm does not lead to an optimal enumeration of Eulerian trails in directed graphs, we were motivated to develop a direct algorithm for this problem. Interestingly, the algorithm by Conte et al.~can be lifted to an enumeration algorithm with some effort~\cite[Section 5.1]{TKDD2025}, but it still requires $\Theta(m\cdot \totalETs)$ time. 

Our central contribution is a remarkably simple algorithm to \emph{directly} enumerate Eulerian trails in directed graphs in the \emph{optimal} $\cO(m + \totalETs)$ time.\footnote{An implementation of this algorithm written by Jeroen op de Beek is made available at \url{https://github.com/Jeroenodb/Optimal-Enumeration-of-Euler-Trails}.} Thus, we complete the complexity landscape of computing Eulerian trails in undirected and directed graphs; see \Cref{tab:landscape}.
Furthermore, our algorithm has the following important consequences:
\begin{itemize}[noitemsep]
    \item It improves on an implementation of the BEST theorem for counting Eulerian trails in directed graphs when $\totalETs=o(n^2)$---only to \emph{construct} the underlying Laplacian matrix, $\Omega(n^2)$ time is required.
    \item It improves the combinatorial $\cO(m\cdot \totalETs)$-time algorithm of Conte et al.~\cite{TKDD2025} for counting Eulerian trails in directed graphs. Like the algorithm of Conte et al., our algorithm is also combinatorial.
\end{itemize}

\begin{table}[t]
  \centering
\begin{tabular}{lllll}
    \toprule
 & \multicolumn{2}{c}{\textbf{Undirected}} & \multicolumn{2}{c}{\textbf{Directed}}  \\ \cmidrule(r){2-3}\cmidrule(l){4-5}
 & Bound & Reference & Bound & Reference \\ \midrule
\textbf{Decision} & $\cO(m)$ time & Euler's theorem  & $\cO(m)$ time & Euler's theorem  \\
\textbf{Witness} & $\cO(m)$ time & \cite{hierholzer} & $\cO(m)$ time & \cite{hierholzer}  \\
\textbf{Counting} & $\textsf{\#P}$-complete &\cite{DBLP:conf/alenex/BrightwellW05} &  $\cO(\textrm{poly}(m))$ time &BEST theorem\\
\textbf{Enumeration} & $\cO(m + \totalETs)$ time & \cite{DBLP:journals/tcs/KuritaW22}  & $\cO(m + \totalETs)$ time &\textcolor{CWIRed}{[This paper]} \\
\bottomrule
\end{tabular}
\caption{Complexity landscape of computing Eulerian trails in undirected and directed graphs.}
\label{tab:landscape}
\end{table}

The algorithm by Kurita and Wasa~\cite{DBLP:journals/tcs/KuritaW22} for the optimal
enumeration of Eulerian trails in undirected graphs
is based on \emph{reverse search}, a widely-used technique that was introduced in 1996 by Avis and Fukuda~\cite{DBLP:journals/dam/AvisF96}.
We use the same overall approach for directed graphs, but introduce combinatorial insights to avoid the explicit connectivity computation at every step.
In particular, in the undirected case, one crucially needs to track bridges (edges between 2-connected components).
In the directed case, a similar role is played by crossing edges (edges between strongly connected components).
However, there is strong evidence from fine-grained complexity~\cite{6979028} that we cannot dynamically keep track of the latter efficiently enough for our purposes.
In response, we develop combinatorial insights showing that tracking an easier-to-detect subset is sufficient in the context of Eulerian trail enumeration. 
This overcomes the challenges in adapting the algorithm in~\cite{DBLP:journals/tcs/KuritaW22} to directed graphs. Additionally, this avoids the intricate \emph{push-out amortization} technique, which was introduced by Uno~\cite{DBLP:conf/wads/Uno15} in 2015 and used by Kurita and Wasa in~\cite{DBLP:journals/tcs/KuritaW22}, making our algorithm both simple and self-contained. Crucially, avoiding this amortization tool allows us to output a desired number $z\leq \totalETs$ of Eulerian trails in the optimal $\cO(m+z)$ time.

Since outputting $z$ Eulerian trails explicitly requires $\Omega(mz)$ space and time, we output a \emph{compressed} representation: a rooted tree, in which each node represents a prefix of an Eulerian trail, and some inner nodes having a single child may be omitted, and are thus implicit. These implicit nodes represent prefixes of Eulerian trails that can be extended only by a unique edge of the input graph, making the decompression of the output natural---in fact, the same tree representation has been used by Kurita and Wasa~\cite{DBLP:journals/tcs/KuritaW22}.\footnote{This tree representation is standard in enumeration algorithms. Compare it to how the suffix tree of a string is obtained from the trie of its suffixes~\cite{DBLP:conf/focs/Weiner73}.}

In addition to the rich history of theoretical interest in this problem, Eulerian trails have direct applications in bioinformatics~\cite{DBLP:journals/bmcbi/KingsfordSP10} and data privacy~\cite{DBLP:journals/jea/BernardiniCFLP21}. These applications typically operate on directed multigraphs, which motivated us to extend our algorithm to this setting, thus opening the possibility of boosting the efficiency of these applications compared to the existing algebraic computations that are based on the BEST theorem (see~\cite{TKDD2025} for details). Indeed, our direct approach extends to directed multigraphs with little additional effort.

\subsection{Technical Overview and Paper Organization}
\subparagraph{The State Tree.} Assume a directed graph $G$ has $\totalETs>0$ Eulerian trails starting at some fixed node $v_0$ (this restriction to a single start node can easily be removed, but it is conceptually easier). Consider the set of prefixes of these trails.
They can be arranged into a tree where there is an edge between two prefixes if one extends the other by a single edge.
Note that the set of leaves is exactly the set of Eulerian trails.
See \Cref{fig:uncompressed-tree} for an example.
Since every Eulerian trail has length $m$, this tree has size $\cO(m \cdot \totalETs)$.
Since we construct this tree by recursively exploring it, the current node in the tree encodes the state our exploration algorithm is in.
We refer to this tree as the \emph{state tree}, a node as a \emph{state}, and an edge as a \emph{state transition} to avoid any confusion with the nodes and edges from $G$. See \Cref{sec:ST}.
\subparagraph{A Simple $\cO(m \cdot \totalETs)$-time Enumeration Algorithm.}
A first observation is that this full state tree can be constructed by recursively completing prefixes using Hierholzer's algorithm for finding an Eulerian trail.
Given some inner state (i.e., a proper prefix of an Eulerian trail), a single call to Hierholzer's algorithm completes the prefix into an Eulerian trail: it creates a path to a leaf state from that inner state.
Exploiting a simple fact, we can detect which of the newly created inner states are branching, meaning their prefixes can be completed to a different Eulerian trail.
All of these branching states can then be explored recursively in the same fashion.
This already yields an $\cO(m \cdot \totalETs)$-time algorithm, matching the state-of-the-art algorithm by Conte et al.~\cite{TKDD2025}, but with a much simpler algorithm. See \Cref{sec:mz}.
\subparagraph{Compression: Avoiding Repetition in the State Tree.}
The core idea that enables us to obtain a time-optimal enumeration algorithm is that 
Eulerian trails are repetitive in a way that can be exploited to obtain a more compact representation.
States in the state tree with only one child correspond to prefixes that can be extended in exactly one way.
So, in some sense, they do not encode useful information.
One can observe this in \Cref{fig:uncompressed-tree}.
Thus, our goal is to compress the state tree so that every state has at least two child states (in which case, we call it \emph{branching}).
Unfortunately, we cannot quite achieve this. 
Instead, we settle for a slightly weaker result showing that, with the right compression, there can be only a constant number of non-branching states in a row.
This is enough to establish that the compressed state tree has size $\cO(\totalETs)$. This gives the desired output. See \Cref{sec:comp}.
\subparagraph{Our Compression Rules.}
We propose two compression rules.
The first is simple: if a node has exactly one out-edge, then on every Eulerian trail, that edge must be used as soon as that node is reached.
Thus, we can contract the node with its out-neighbor without changing the number of Eulerian trails. See \Cref{sec:comp1}.

The second compression rule exploits that if there is a node that has exactly one out-edge to its own strongly connected component (SCC) and exactly one edge to a different SCC, the edge to the same SCC must be used \emph{first} on all Eulerian trails. Thus, the node can be contracted with its out-neighbor from the same SCC. Unfortunately, exhaustively applying this rule after every state transition seems difficult (and there is evidence from fine-grained complexity to support this). To overcome this difficulty, we propose a weaker (but more technical) rule. See \Cref{sec:comp2}. This compression rule suffices to bound the size of the (compressed) state tree by $\cO(m+\totalETs)$. See \Cref{sec:size}.
\subparagraph{Efficient Exploration.}
Our time-optimal algorithm constructs the state tree similarly to the $\cO(m \cdot \totalETs)$-time algorithm outlined above. See \Cref{sec:explore}.
There is an extra algorithmic challenge: we need to maintain what we call the \emph{compressed remaining graph}: while exploring the state tree we want to have access to $G$ after the Eulerian trail prefix corresponding to the current state has been removed. We also want this remaining graph to be exhaustively compressed using our compression rules. By the simplicity of the compression rules, one compression is easy to execute, so the important fact here is that compressions cannot cascade so that many compressions become necessary at one time. See \Cref{sec:remain}. As in the $\cO(m \cdot \totalETs)$-time algorithm, we will iteratively use Hierholzer's algorithm to extend the partially constructed state tree. See \Cref{sec:Hier}.
\subparagraph{Putting It All Together.}
We combine the above ingredients to form our algorithm. See \Cref{sec:algo}.
Note that, if we want to enumerate $z$ Eulerian trails in $\cO(m+z)$ time, where $z$ is an arbitrary input parameter, 
we can terminate the construction of the state tree early.

The remainder of this paper is organized as follows.
We start in \Cref{sec:prel} with the necessary definitions and notation.
In \Cref{sec:multi}, we show how our algorithm for simple graphs can be extended, with some care, to enumerate Eulerian trails in directed multigraphs in optimal time.
\Cref{sec:final} concludes the paper.

\section{Preliminaries}\label{sec:prel}
A \emph{strongly connected component} (SCC) in a directed graph $G$ is a maximal subgraph in which every node can reach every other node using a directed path. The SCCs in $G$ form a partition of the set of nodes. Thus, without any ambiguity, for any node $v$ in $G$, we write $\scc{v}$ to refer to the SCC of $v$ in $G$.

A \emph{crossing} in a directed graph $G$ is an edge between two different SCCs.
This is consistent with the usage for any edge between two different parts of a partition. In our case, the partition will always be given by the SCCs in $G$.\footnote{Sometimes, this is called a (weak) bridge. However, since the term bridge is used inconsistently in directed graphs, we use the term \emph{crossing} to avoid any confusion. Additionally, intuitions about undirected bridges do not carry over. For our purposes, they can even be counterproductive.}

A \emph{contraction} of two nodes $v$ and $w$ in a directed graph $G$ modifies the graph as follows.
The nodes $v$ and $w$ are removed and a new node $vw$ is added.
This new node has all in-neighbors from $v$ and $w$ as in-neighbors and all out-neighbors from $v$ and $w$ as out-neighbors.
We create a self-loop on the new node $vw$ if and only if both edges $v \to w$ and $w \to v$ exist. This choice is motivated by our application.

We write $z$ for the target number of Eulerian trails to enumerate in a directed graph $G$; and we write $\totalETs(G)$ for the total number of Eulerian trails in $G$. The graph $G$ may be omitted if it is clear from the context.

\section{The State Tree}\label{sec:ST}
The state tree, as it is commonly used in the reverse search framework, is both our core data structure and the compressed output of our main algorithm.
This is the same basic data structure used in the algorithm by Kurita and Wasa for undirected graphs~\cite{DBLP:journals/tcs/KuritaW22}. By definition, explicitly outputting $z$ Eulerian trails requires \(\Omega(mz)\) time and space.
Our main contribution is to construct a tree of size \(\cO(\totalETs)\), where every leaf corresponds to one of the $\totalETs$ Eulerian trails in $G$.
We also show how to explore a subtree of the state tree 
having $z$ leaves in $\cO(m+z)$ time.

In the rest of the paper, we assume the input graph $G=(V,E)$ is a simple and weakly connected 
directed graph with at least one Eulerian trail.
As mentioned in \Cref{sec:intro}, the latter can be verified in $\cO(m)$ time using Euler's theorem.
A witness can then be found in $\cO(m)$ time using Hierholzer's algorithm~\cite{hierholzer}.

Recall that either there is a unique source node $s$ for all Eulerian trails (i.e., a node $s:=v$ with $\outdegree{v} = \indegree{v} + 1$) or the graph contains an Eulerian cycle.
If there is an Eulerian cycle, all Eulerian trails can be rotated to start with an arbitrary fixed node. Henceforth, we assume a fixed source node $v_0$.

We define the \emph{uncompressed state tree} of $G$, denoted by $T^u$, as a rooted tree, where each node of the tree represents a prefix of an Eulerian trail starting at $v_0$.
See \Cref{fig:uncompressed-tree} for an example.
We add an edge from a tree node $s_1$ to another tree node $s_2$ if $s_1$ corresponds to the prefix obtained by removing the last edge from the prefix corresponding to $s_2$. In the rest of the paper, to avoid confusion, we reserve the terms \emph{node} and \emph{edge} for the input graph $G$, and use \emph{state} and \emph{transition} for the (uncompressed) state tree. For a state $s$, we write $\prefix{s}$ for the prefix (of an Eulerian trail) corresponding to $s$.
The following observation is then trivial.

\begin{observation}
    \label{obs:uncompressed-tree-size}
    The size of $T^u$ constructed over an $m$-edge graph $G=(V,E)$ is in $\cO(m \cdot \totalETs)$.
\end{observation}

\begin{figure}[htbp]
    \centering
    \includegraphics[page=13]{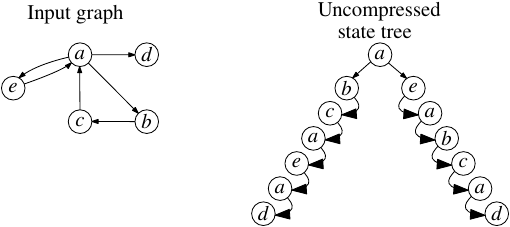}
    \caption{A directed graph and the uncompressed state tree encoding the two Eulerian trails starting at $a$. For readability, each state is only labeled with the last node on the prefix.}
    \label{fig:uncompressed-tree}
\end{figure}

\section{Warm-up: A Simple $\cO(m\cdot\totalETs)$-time Enumeration Algorithm}
\label{sec:mz}

As a warm-up, we describe a slower, $\cO(m\cdot \totalETs)$-time version of our main algorithm. We outline our algorithmic strategy and provide an intuition on why even this simple version ties the state-of-the-art algorithm by Conte et al.~\cite{TKDD2025}. The full proof for this version is omitted, as we later prove the faster version in detail. The rest of the paper refines this simple algorithm to achieve the optimal $\cO(m+\totalETs)$-time complexity.

We construct the state tree in a DFS fashion.
Given a prefix of an Eulerian trail in $G$, we can leverage Hierholzer's algorithm~\cite{hierholzer} to complete it into a full Eulerian trail in $G$.
In this way, given an inner state $s$ (potentially the initial state), we construct a path from $s$ to a new leaf. 
We repeat this until we have constructed a part of the state tree with $z$ leaves (or have explored the entire state tree).

The remaining challenge is to find out which of the states on the path from $s$ to a leaf, constructed by Hierholzer's algorithm, are \emph{branching} (i.e., have more than one child). At these branching states, we must call Hierholzer's algorithm again with a modified prefix to construct the subsequent set of states.

Recall that a crossing is an edge between two SCCs.
This gives \Cref{lem:crossings}, where by $\remainingGraph{s}$ we denote the \emph{remaining graph} after deleting the prefix $\prefix{s}$ of an Eulerian trail from $G$ and then removing isolated nodes.

\begin{lemma}\label{lem:crossings}
    A state $s$ of $T^u$ is branching if and only if the last node on $\prefix{s}$ has at least two out-edges that are not crossings in $\remainingGraph{s}$.
\end{lemma}

Intuitively, only non-crossing out-edges correspond to real choices because: 
(1) non-crossing edges must be used before crossing edges; (2) there can be at most one crossing edge at each node; 
and (3) any non-crossing edge is a valid extension of a prefix.
See \Cref{fig:non-crossing-are-choices} for an illustration.

\begin{figure}[tbhp]
    \centering
    \includegraphics[page=12]{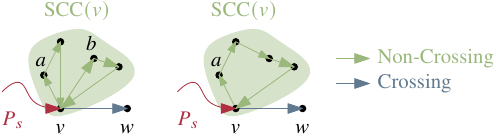}
    \caption{On the left, $\prefix{s}$ can be extended both by $v \to a$ and by $v\to b$ since both of these out-edges are non-crossing. On the right, $\prefix{s}$ can be extended only by $v \to a$. If instead we choose the crossing edge $v \to w$ first, the other edges from $\scc{v}$ would not be reachable anymore.}
    \label{fig:non-crossing-are-choices}
\end{figure}

\begin{proof}[Proof of \Cref{lem:crossings}]
    Let us prove that (1), (2), and (3) are true.
    The lemma then follows.

    For (1), assume $v$ is the last node on $\prefix{s}$ and that, in $\remainingGraph{s}$, it has both a crossing out-edge $v \to w$ and a non-crossing out-edge $v \to a$.
    Assume for contradiction that $\prefix{s} \cdot (v \to w)$ is a prefix of an Eulerian trail.
    Then there is some way to extend it such that $v \to a$ later appears on it.
    This means that $w$ lies on a cycle from $v$ to $v$, and therefore $w$ is in $\scc{v}$.
    So, $v \to w$ was not crossing.
 
    For (2), the argument is similar. Assume for contradiction that $v$ has two crossing out-edges in $\remainingGraph{s}$, $v \to w$ and $v \to w'$. An Eulerian trail must traverse every edge in $\remainingGraph{s}$ exactly once. If the trail leaves the SCC of $v$ via $v \to w$, it can never return to $v$ because, by definition, a crossing connects to a different SCC with no path back. This means the edge $v \to w'$ remains untraversed, which is a contradiction.

    For (3), removing a non-crossing out-edge adjacent to the last node on $\prefix{s}$ (a) maintains Euler's degree criterion and (b) leaves the graph weakly connected. By Euler's theorem, this guarantees the existence of an Eulerian trail in $\remainingGraph{s}$, and thus the extended prefix is valid.
\end{proof}

With \Cref{lem:crossings}, we can efficiently detect branching states. 
Crucially, we do not need to maintain the complete set of crossing edges (which might be difficult as we explain in \Cref{sec:comp2}).
Instead, given a state $s$, we only need to know which out-edges are crossing at the last node on $\prefix{s}$.
As it turns out, this is easy.

\begin{lemma}
    \label{lem:incremental-scc}
    Given a graph $H$ and an Eulerian trail $R = e_1, \dots, e_k$ in $H$, there is an algorithm running in $\cO(\setsize{R})$ time that determines which of the edges $e_i$ are crossings in $H \setminus \{e_1, \dots, e_{i-1}\}$.
\end{lemma}

\begin{proof}
    The crossings are precisely the last out-edges at each node, with respect to their order in $R$---except if $R$ is a cycle, in which case the last out-edge at the end node of the cycle is not a crossing.
    The last out-edge from $v$ must be a crossing because all other out-edges from $v$ have been used, and to complete the trail, the algorithm must leave the current SCC (unless the trail is about to end).
    A non-last out-edge $v\to w$ cannot be a crossing because the trail must eventually return to $v$ to use its other out-edges. The last edge of a cycle is: it is the last edge, but it returns to the start node, which is in the same SCC.
\end{proof}

This enables us to detect the branching states on a path from an inner state to a leaf in time linear in the size of that path.
For a path from an inner state $s_1$ to a leaf $s_2$, we apply \Cref{lem:incremental-scc} to $H \coloneqq \remainingGraph{s_1}$ and $R \coloneqq \prefix{s_2} - \prefix{s_1}$. Combining this insight with the DFS strategy based on Hierholzer's algorithm shows that we can construct the state tree in time linear in the number of states we construct. 
This allows us to avoid the application of the push-out amortization technique.
Together with \Cref{obs:uncompressed-tree-size}, which bounds the size of the (uncompressed) state tree, this gives the following result.

\begin{theorem}
    There is an algorithm that, for any directed $m$-edge graph $G=(V,E)$ and any parameter $z \in \N$, computes $z$ Eulerian trails in $G$ (or all $\totalETs$ of them if $\totalETs\leq z$) in $\cO(m\cdot\min(z,\totalETs))$ time.
\end{theorem}

We recovered the state-of-the-art result~\cite{TKDD2025} with a simpler algorithm. Although the algorithm in~\cite{TKDD2025} is for assessing $\totalETs$, it can be lifted to enumerate the trails with some effort~\cite[Section 5.1]{TKDD2025}, and it also works for directed multigraphs.

\section{Compression: Avoiding Repetition in the State Tree}\label{sec:comp}
Notice that $T^u$ can be highly repetitive. For example, in \Cref{fig:uncompressed-tree}, the cycle $a \to b \to c \to a$ in the input graph appears in both Eulerian trails from node $a$ to node $d$, so the left and right subtrees in the state tree share duplicate information. We aim to compress $T^u$ such that there are only branching states (states with at least two children) and leaves. We then rely on the following standard fact to bound the number of branching states in terms of the number of leaves (and thus in terms of the number of Eulerian trails). 

\begin{fact}
    \label{fac:limit-branching-states}
    A rooted tree with $k$ leaves has at most $k$ branching states.
\end{fact}

Unfortunately, having only branching states and leaves is not (efficiently) attainable. Instead, we settle for the more modest goal of having only a constant number of non-branching inner states per branching state.
Concretely, we bound the number of non-branching states on a path between two branching states. For our compression, we take a closer look at when non-branching states can occur. If $s_1$ is non-branching (with child state $s_2$), we call the transition $s_1 \to s_2$ \emph{forced} because it is the only way to extend the prefix $\prefix{s_1}$.

\subsection{\Cref{rule:1}: Only One Out-edge}\label{sec:comp1}

The simplest compression case concerns nodes with out-degree one.
Given a state $s$, if there is a node $v$ in the remaining graph $\remainingGraph{s}$ with out-degree one, it is clear that any completion of the prefix $\prefix{s}$ must use this out-edge as soon as it reaches $v$. 
Note that for any state $s$, the last node on $\prefix{s}$ (or just \emph{last node}, for short) is the node where the next edge has to depart from.
Thus, if $w$ is the only out-neighbor of $v$, then in the subtree of the state tree rooted at $s$, any state $s_1$ with current node $v$ has a single child state $s_2$ with current node $w$.
Thus, the transition $s_1 \to s_2$ is forced.
For an example, recall \Cref{fig:uncompressed-tree}.
There, nodes $b$, $c$, and $e$ have only a single out-edge, resulting in forced transitions.

\begin{compression}
    \label{rule:1}
    For any state $s$, if in the remaining graph $\remainingGraph{s}$, there is a node $v$ with out-degree one, contract that node and its only out-neighbor. 
\end{compression}

\begin{figure}[th]
    \centering
    \includegraphics[page=11]{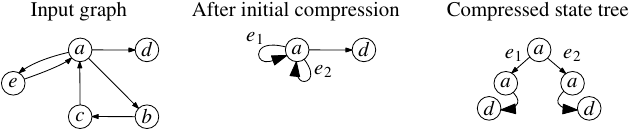}
    \caption{The graph from \Cref{fig:uncompressed-tree} after \Cref{rule:1} has been applied exhaustively, together with the compressed state tree. Compare this smaller state tree to the one from \Cref{fig:uncompressed-tree}.}
    \label{fig:compressed-tree}
\end{figure}

Since contractions are associative and commutative (up to isomorphism), we do not need to specify their order. Note that after applying this compression rule, the graph may not be simple anymore; we allow each node to have one or more self-loops.
To visualize the effect of this rule, see \Cref{fig:compressed-tree}. 
The next lemma proves its soundness.

\begin{lemma}
    For any state $s$, there is a bijection between Eulerian trails in the remaining $\remainingGraph{s}$ and the same graph after the application of \Cref{rule:1}.
\end{lemma}

\begin{proof}
    The claimed bijection $f$ is given by the contractions performed.
    To see the claim, take any edge $v\to w$ that is contracted by \Cref{rule:1}.
    Then the out-degree of $v$ in $\remainingGraph{s}$ is one and thus in any Eulerian trail in $\remainingGraph{s}$, $v$ must be directly followed by $w$.
    This proves that $f$ is injective.
    To see that $f$ is surjective, observe that any Eulerian trail in the compressed graph can be extended to an Eulerian trail in the uncompressed graph by undoing the contractions.
\end{proof}

In our algorithm, we will explore the state tree while efficiently maintaining the compressed remaining graph for the current state.
More formally, we have the following lemma.

\begin{lemma}
  \label{lem:rule-1-cascade}
    Let $s_1 \to s_2$ be a transition in the state tree that appends the edge $v \to w$ to $\prefix{s_1}$.
    Assume $\remainingGraph{s_1}$ is fully compressed with \Cref{rule:1}. When generating $\remainingGraph{s_2}$, \Cref{rule:1} needs to be applied at most once (specifically at $v$).
\end{lemma}

\begin{proof}
    Since $\remainingGraph{s_2}$ is obtained from $\remainingGraph{s_1}$ by deleting the edge $v \to w$ and all isolated nodes, the only node at which the out-degree decreases (and thus may become one) is $v$.
\end{proof}

\subsection{\Cref{rule:2}: Crossings}\label{sec:comp2}
By the definition of crossing, a prefix of an Eulerian trail that ends at a node $v$ that has a crossing $v \to w$ must first be extended with the remaining edges in the SCC of $v$ before using the crossing.
Indeed, this follows because after the crossing, the trail cannot return to the SCC of $v$, and therefore cannot use the remaining edges within that component.
See \Cref{fig:possible-scc-compression} for an example.
This leads to the following observation.

\begin{figure}[tbhp]
    \centering
    \includegraphics[page=15]{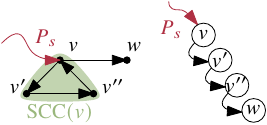}
    \caption{For a state $s$ in the state tree, $\prefix{s}$ is some prefix of an Eulerian trail. On the left, notice that $v \to w$ is a crossing. Thus, $\prefix{s}$ can only be extended by first using the edges from $\scc{v}$ and then using $v \to w$. In the state tree, displayed on the right, the state $s$ therefore only has a single child state.}
    \label{fig:possible-scc-compression}
\end{figure}

\begin{observation}[Candidate Compression Rule]
    \label{obs:forced-crossings}
    Let $s$ be a state and let $v$ be the last node on its prefix $\prefix{s}$.
    If $v$ has at most one out-edge that is not a crossing, $s$ has only one child state.
\end{observation}

We would like to turn \Cref{obs:forced-crossings} into a compression rule.
This would even be a stronger version of \Cref{rule:1} and would immediately give us that the compressed state tree contains only a few non-branching states.
For our algorithm, it is crucial that after each state transition, we can efficiently maintain the compressed remaining graph.
Unfortunately, this seems difficult for \Cref{obs:forced-crossings}.\footnote{We effectively want a decremental algorithm maintaining strongly connected components with undo support. Unfortunately, the best algorithms give only amortized guarantees for the decremental case and any fast enough algorithm with worst-case guarantees would violate a conditional lower bound. Technically, these general lower bounds do not rule out an algorithm for the graph class of Eulerian graphs, but there is no known decremental algorithm amenable to the undo support we would require. See~\cite{DBLP:journals/siamcomp/BernsteinGW19,DBLP:conf/focs/BernsteinGS20} for the currently best decremental SCC algorithms and~\cite{6979028} for the conditional lower bounds.}
Instead, we will use the following.

\begin{compression}
    \label{rule:2}
    Let $s$ be any state. If there is a node $v$ in $\remainingGraph{s}$ with only one self-loop, exactly one other out-edge, and exactly one other in-edge, remove the self-loop.

    If $v$ is the last node on $\prefix{s}$, we remove the self-loop only if $v$ has no other in-edge.
\end{compression}

At first, this may appear like a restrictive special case (and it is), but we will see that (together with \Cref{rule:1}) it guarantees that in a compressed graph, instances where \Cref{obs:forced-crossings} applies are rare (and do not occur in a row).
See \Cref{fig:type-2-compression} for an example of how these rules act together.

\begin{figure}[tbhp]
    \centering
    \includegraphics[page=16]{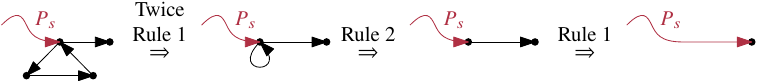}
    \caption{Recall the graph from \Cref{fig:possible-scc-compression}.
    We can see that together \Cref{rule:1,rule:2} compress this graph so that no forced state transition remains.}
    \label{fig:type-2-compression}
\end{figure}

 Before we proceed with proving the size bound, we give \Cref{lem:CR2-sound}, which proves that \Cref{rule:2} is sound.

\begin{lemma}\label{lem:CR2-sound}
    For any state $s$, there is a bijection between Eulerian trails in the remaining $\remainingGraph{s}$ and the same graph after the application of \Cref{rule:2}.
\end{lemma}

\begin{proof}
    The claimed bijection $f$ is given by the self-loop removals.
    For simplicity, assume a single self-loop $v \to v$ was removed.
    The result then extends to multiple rule applications by chaining the resulting bijections.

    For injectivity, assume that we have two Eulerian trails $T_1\neq T_2$ with $f(T_1)=f(T_2)$. The difference between $T_1$ and $T_2$ cannot be due to the self-loop: its position is \emph{fixed} because the out-edge at $v$ must be a crossing.
    This holds as the condition on the number of in-edges at $v$ means that $v$ cannot be visited again after using the out-edge.
    Thus, the self-loop must be used directly before the out-edge.
    Also, the self-loop must be traversed exactly once in both trails.
    The difference must therefore lie in a different sequence of edges, which would mean that $f(T_1)\neq f(T_2)$, contradicting our initial assumption.
    Thus, $T_1\neq T_2$ is false.

    For surjectivity, again, observe that any Eulerian trail in the compressed graph can be extended to an Eulerian trail in the uncompressed graph by undoing the self-loop removal. Let $T'$ be any Eulerian trail in the compressed graph.
    Then, since $v$ has an in-edge in the compressed graph by the condition of \Cref{rule:2}, the trail $T'$ must pass through $v$.
    By inserting the self-loop $v \to v$ at this position in $T'$, we can construct an Eulerian trail for the uncompressed graph because the degrees of $v$ remain balanced.
\end{proof}

There is exactly one structural scenario where applications of both rules can cascade into each other, but a slight adjustment to the graph representation can easily avoid this.

\subsection{Handling Induced Two-Way Paths}
\label{sec:cascade-fix}
In particular, a cascade can occur along an \emph{induced two-way path}: a sequence of nodes and edges $v_1 \rightleftarrows \dots \rightleftarrows v_k$, where the inner nodes $v_2, \dots, v_{k-1}$ have exactly in-degree 1 and out-degree 1 in both directions. 
We represent such a path by replacing the inner nodes with a special meta-node annotated with the length $k$ \ie{a constant-size representation}.
If an induced two-way path becomes disconnected from the rest of the graph at one of its endpoints \ie{only one of $v_1, v_k$ has an edge other than those of the induced two-way path}, the entire path can be immediately contracted into a self-loop on the remaining connected endpoint.
When running Hierholzer's algorithm, if an induced two-way path is entered from both ends, we can easily generate all branching states corresponding to the valid prefix and suffix pairs of that path.
If the path is only entered from one end, all state transitions traversing the implicitly represented inner nodes are non-branching and can be safely omitted from the explicit state tree.
We can handle these updates efficiently.

\begin{lemma}
    Let $s_1 \to s_2$ be a transition in the state tree that appends the edge $v \to w$ to $\prefix{s_1}$.
    Then, any induced two-way paths can be updated in $\cO(1)$ time, and \Cref{rule:1} will be applied at most once while updating the remaining graph.
    \label{lem:rule-1-2-cascade}
\end{lemma}

\begin{proof}
    First, observe that updating the induced two-way paths is easy.
    The only things that can happen are that $v$ leaves one, enters one, or two of them merge at $v$. All of these things can be checked and performed in $\cO(1)$ time. As \Cref{rule:1} was not applicable at $w$ before, the only case in which it can become applicable is if it now cascades through an induced two-way path that is only connected at one end.
    This follows as an induced two-way path is the only case where contracting an edge $u \to v$ not only reduces the in-degree of $v$, but also its out-degree (after an application of \Cref{rule:2}), making \Cref{rule:1} applicable at $v$.
    The induced property then follows from the degree-criterion for Eulerian graphs.
    By the checks we just performed, it would have been turned into a self-loop. Thus, no cascading happens.
\end{proof}

\subsection{The Size of the Compressed Graph with Rules 1 and 2}\label{sec:size}
The core result that underlies our algorithm is that since in each state we compress the remaining graph using \Cref{rule:1,rule:2}, the resulting graph has few non-branching nodes. This follows from \Cref{lem:high-branching}.

\begin{lemma}
    \label{lem:high-branching}
    Assume that in the state tree \Cref{rule:1,rule:2} are applied exhaustively after each transition.
    Then there are never two forced state transitions in a row (i.e., there is no path of length two of forced transitions).
\end{lemma}

\begin{proof}
    Assume to the contrary that $s_1, s_2, s_3$ are three states such that the two transitions $s_1 \to s_2$ and $s_2 \to s_3$ are forced.
    Then $s_1$ and $s_2$ only have a single child state each.
    For $i\in [1,3]$, write $v_i$ for the last node on $\prefix{s_i}$.
    So, by assumption, appending first $v_2$ and then $v_3$ to $\prefix{s_1}$ is the only way to extend $\prefix{s_1}$ keeping it a prefix of an Eulerian trail.
    Let our analysis proceed in $\remainingGraph{s_1}$.
    This graph is Eulerian because removing a prefix of an Eulerian trail does not change that property. 
    Nodes $v_1$ and $v_2$ have out-degree at least two as otherwise \Cref{rule:1} would be applicable.
    Notice that they have at most one non-crossing out-edge as otherwise the states would be branching.
    Thus, we can conclude that $v_1$ and $v_2$ each have exactly one non-crossing out-edge and one crossing out-edge.
    The edges $v_1 \to v_2$ and $v_2 \to v_3$ must be the non-crossing edges as those must appear before the crossings on any Eulerian trail.

\begin{figure}[tbhp]
\centering
        \includegraphics[page=6]{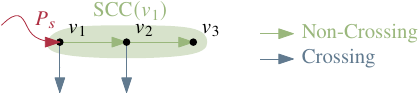}
        \caption{The contradictory situation created in the proof of \Cref{lem:high-branching}. Here, the only way to extend $P_s$ is by making two forced choices in a row (adding $v_1 \to v_2$ and $v_2 \to v_3$).
        In the proof we argue why this cannot happen if $G$ is Eulerian and compressed using \Cref{rule:1,rule:2}.
        In particular, no SCC can have two outgoing crossings.
        }
        \label{fig:high-branching-proof}
    \end{figure}
    None of these edges can be self-loops as otherwise \Cref{rule:2} would have been applicable.
    Thus, $v_1$ and $v_2$ are distinct nodes.
+    Since they are connected by a non-crossing edge, they must be in the same SCC.
    This SCC must now have two outgoing crossings.
    This can never occur in any Eulerian graph (also see~\cite[Lemma 1]{TKDD2025} for a more general version of this observation). 
    See \Cref{fig:high-branching-proof} for an illustration.
\end{proof}

This lemma, together with \Cref{fac:limit-branching-states}, yields the following main result.

\begin{theorem}
    \label{thm:compressed-size}
    The state tree compressed using \Cref{rule:1,rule:2} has size $\cO(\totalETs)$.
    In particular, any state with $z$ leaf-descendants has $\cO(m+z)$ ancestors and descendants.
\end{theorem}

\section{Enumeration Algorithm}\label{sec:explore}
We have established that the compressed state tree is small enough. 
Our next step is to show how to efficiently construct it via a DFS exploration, thereby enumerating all Eulerian trails. 
This approach requires two main ingredients: (1) an efficient way to maintain the compressed remaining graph for the current state; 
and (2) a method to extend the partial state tree.
In the usual DFS fashion, we would add new states one-by-one, but we need to be more clever and batch some operations.
Concretely, we will extend the state tree by adding a full path of states from the current state to a new leaf of the state tree at a time.
Crucially, we make sure that the remaining graph for each state is always exhaustively compressed.

\subsection{Keeping the Remaining Graph Updated and Compressed}\label{sec:remain}
The main observation towards efficiently computing a compressed state tree is that removing a single edge only triggers a constant number of compression operations.
In particular, compression operations do not cascade, that is, a single compression
does not trigger a chain reaction that creates many new compressions.
\Cref{lem:fast-compression} formalizes this.

\begin{lemma}
Given a graph $G$, exhaustively compressed by \Cref{rule:1,rule:2}, and an edge $e$, we can compute $G-e$, after \Cref{rule:1,rule:2} have been applied exhaustively, in $\cO(1)$ time.
\label{lem:fast-compression}
\end{lemma}
\begin{proof}
    Let $e=v\to w$. Removing $e$ decreases the out-degree of exactly one node, namely $v$.
    Thus, \Cref{rule:1} can become applicable only at $v$ (or nowhere).
    All resulting applications of the rule are handled in $\cO(1)$ time by \Cref{lem:rule-1-cascade,lem:rule-1-2-cascade}.

    Similarly, \Cref{rule:2} can become applicable only at $v$ or $w$.
    At each of these, it can only be applied once (since by the condition of \Cref{rule:2} they have exactly one self-loop).
    At each of these nodes, \Cref{rule:1} can become applicable, but nowhere else.
    As argued above, potentially applying \Cref{rule:1} does not cause further compressions.
    Finally, note that checking the conditions of the rules at $v$ and $w$, the contraction of two nodes and the removal of a self-loop can be implemented in $\cO(1)$ time.
\end{proof}

This implies that the state tree can be explored in a DFS fashion while keeping an up-to-date compressed version of the remaining graph with constant overhead per state transition. Note that by storing the performed compression operations on a stack, they can be undone in the same constant overhead per state transition while the DFS exploration backtracks.
Additionally, since by the above arguments the compression does not cascade, the following holds.
\begin{corollary}
    \label{rem:initial-compression}
    We can exhaustively compress any directed $m$-edge graph in $\cO(m)$ time.
\end{corollary}

\subsection{Exploring Using Hierholzer's Algorithm}\label{sec:Hier}

As in \Cref{sec:mz}, we will iteratively use Hierholzer's algorithm to extend a partially constructed state tree.
However, this time the state tree will be compressed using \Cref{rule:1,rule:2}.
Essentially, we will perform a DFS, where Hierholzer's algorithm will serve as the subroutine that adds a new path from a not-fully-explored inner state to a new leaf. To obtain a time-optimal algorithm, we must analyze the cost of these calls more carefully than in \Cref{sec:mz}, where their cost was linear in the number of states in the uncompressed tree. Instead, we will show that the cost is linear in the number of branching states created.

\begin{lemma}
    \label{lem:fast-hierholzer}
    Given a partially constructed compressed state tree and a branching inner state $s$, where one subtree has not been constructed, we can construct a new state-path $R$ from $s$ to a leaf in $\cO(\setsize{R})$ time.
    Here, $\setsize{R}$ is the number of newly created states in the compressed state tree.
\end{lemma}

\begin{proof}
    Note that $R$ corresponds to an Eulerian trail in the remaining graph $\remainingGraph{s}$.
    Hierholzer's algorithm finds this specific trail.
    Our goal is to show that the lengths of this trail and $R$ differ only by at most a constant factor.
    This allows us to conclude that running Hierholzer's algorithm on $\remainingGraph{s}$ takes $\cO(\setsize{R})$ time.
    Informally, we create a number of states that is worth the running time we spent.
    For this, we show that not too much of the trail is removed by compression.

We classify all edges in the Eulerian trail computed by Hierholzer's algorithm in $\remainingGraph{s}$ into three categories: (1) those corresponding to branching states; (2) those corresponding to non-branching states that also exist in the compressed state tree; and (3) those corresponding to non-branching states that are removed by compression.
Clearly, states in categories (1) and (2) are bounded by $\setsize{R}$.

Since $\remainingGraph{s}$ is exhaustively compressed, any states in category (3) must become compressible by the removal of some (earlier) edge on $R$.
In particular, as argued in the proof of \Cref{lem:fast-compression}, an edge in category (1) or (2) can cause only a constant number of possible compressions later on $R$.
In the same proof, we also show that each compression that is caused in this fashion does not cascade. 
As we start off with an exhaustively compressed graph, this is the only way new applications of the compression rules can become possible.
Thus, the number of states in category (3) is bounded by the number of states in categories (1) and (2).

In total, the Eulerian trail constructed has length $\cO(\setsize{R})$ and thus the execution of Hierholzer's algorithm takes $\cO(\setsize{R})$ time, yielding the desired result.
\end{proof}

On a technical note, we need to be a bit careful about picking the state $s$ at which to add a new path to a leaf.
Recall that we may want to output a desired number $z\leq \totalETs$ of Eulerian trails in the optimal $\cO(m+z)$ time.

Formally, if $z \leq \totalETs$, we can construct a subtree with $z$ leaves and $\cO(m+z)$ inner states by adding the new inner state to leaf paths in a DFS fashion; that is, processing a branching state with missing children only after all its already constructed children have no unconstructed children themselves.

\subsection{Summary of the Algorithm}\label{sec:algo}
Despite proving many technical statements---like bounding the state tree size, the soundness of compression rules, and detailing efficient subroutines---the final algorithm is simple to state. See \Cref{alg:summary}.

\begin{algorithm}[tbhp]
    \begin{algorithmic}[1]
    
    \Require{$G=(V, E)$ is weakly connected and Eulerian 
    with source $v_0$, and $z>0$ is an integer parameter.}
        \State $T \gets \{(v_0, [])\}$ \Comment{State tree only containing the empty prefix at $v_0$}
        \State Exhaustively compress $G$ \Comment{\Cref{rem:initial-compression}}

        \While{$T$ has fewer than $z$ leaves}
            \State Take a branching state $s$ with an unexplored subtree (if there is none, break)
            \State Add a path $R$ from $s$ to a new leaf using Hierholzer's algorithm \Comment{\Cref{lem:fast-hierholzer}}
            \State Compress $R$ \Comment{\Cref{lem:fast-compression}}
            \State Mark the branching states with unexplored subtrees on $R$ \Comment{\Cref{lem:incremental-scc}}
        \EndWhile

        \State{\textbf{Output} $T$}

    \end{algorithmic}
    \caption{Computing a subtree of the compressed state tree with $z$ leaves}
    \label{alg:summary}
\end{algorithm}

We also give the properties of this algorithm as a formal statement, summarizing our main contribution.

\begin{theorem}\label{thm:main}
    For any directed $m$-edge graph $G=(V,E)$ and any parameter $z \in \N$, \Cref{alg:summary} computes a compressed representation of $z$ Eulerian trails in $G$ (or all $\totalETs$ of them if $\totalETs\leq z$) in $\cO(m+\min(z,\totalETs))$ time.
\end{theorem}

We stress that the parameterization by $z \in \N$ in \Cref{thm:main}, in conjunction with the $\cO(m+\min(z,\totalETs))$ bound, can be useful in applications~\cite{DBLP:journals/jea/BernardiniCFLP21,TKDD2025}. In particular, these properties of our algorithm  allow one to assess whether $G=(V,E)$ has at least $z\leq \totalETs$ Eulerian trails in $\cO(m+z)$ time, which is useful for real-world graphs with very large $\totalETs$.

\section{Directed Multigraphs}\label{sec:multi}

Counting and enumerating Eulerian trails in directed graphs is an important computational primitive across several domains, such as bioinformatics~\cite{DBLP:journals/bmcbi/KingsfordSP10}, data privacy~\cite{DBLP:journals/jea/BernardiniCFLP21}, and graph compression~\cite{DBLP:conf/kdd/MaserratP10}. For more details, we refer the interested reader to the recent work of Conte et al.~\cite{TKDD2025}.
Yet, in these applications, the input graphs are typically directed \emph{multigraphs}.
Consider the application of investigating genome complexity~\cite{DBLP:journals/bmcbi/KingsfordSP10}. In DNA sequencing, the DNA of a sample is read by sequencing machines in the form of (short) fragments that must then be assembled back into one (or a few) DNA sequences. To this end, order-$k$ de Bruijn multigraphs are usually employed, where Eulerian trails and candidate assemblies are in a one-to-one correspondence~\cite{Pevzner2001}. In particular, the frequency of a multi-edge is the frequency of a DNA $k$-mer (length-$k$ substring), and thus it must be taken into account.
In directed multigraphs, the following two notions of when two Eulerian trails are \emph{distinct} have been proposed~\cite{DBLP:journals/jea/BernardiniCFLP21,TKDD2025,DBLP:journals/tcs/KuritaW22}: edge-distinct and node-distinct.
Two Eulerian trails are \emph{edge-distinct} if their edge sequences are different;
they are \emph{node-distinct} if their node sequences are different.
Note that, in simple graphs, these notions are equivalent, but in multigraphs, this is not the case. See \Cref{fig:edge-v-node-distinct} for an example.
Existing algorithms on directed~\cite{TKDD2025} and undirected~\cite{DBLP:journals/tcs/KuritaW22} multigraphs support both notions.

We show that we can extend our main algorithm (\Cref{thm:main}) to enumerate distinct Eulerian trails under either notion.
This is fairly straightforward for edge-distinctness, but requires some more care for node-distinctness. This is because, for the $\cO(m+z)$ running time, we only count outputted trails in $z$. Thus, if there are multiple edge-distinct Eulerian trails that have the same node sequence, we output it only once and budget only $\cO(1)$ time for it in our running time bound. 
We need to achieve this even though there could be an exponential number of edge-distinct Eulerian trails that all have the same node sequence (as in \Cref{fig:edge-v-node-distinct}).
The node-distinct case is more relevant to biological applications as there, node-distinct trails correspond to distinct DNA assemblies that actually yield distinct strings, while two edge-distinct trails may correspond to two assemblies yielding the same string. In these contexts, users are primarily interested in enumerating assemblies that yield unique strings.
In fact, this is also the case in the data privacy framework proposed in~\cite{DBLP:journals/jea/BernardiniCFLP21}.

\begin{figure}[htbp]
    \centering
    \includegraphics[page=17]{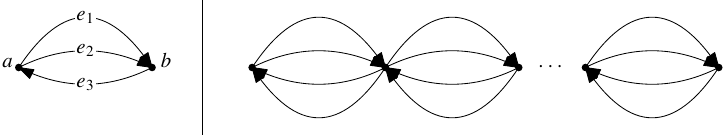}
    \caption{Two directed multigraphs. Left: the trails $e_1 e_3 e_2$ and $e_2 e_3 e_1$ are both Eulerian and edge-distinct, but they are not node-distinct as they both have the same node sequence $a b a b$. Right: a graph class with exponentially many edge-distinct Eulerian trails for the same node sequence.}
    \label{fig:edge-v-node-distinct}
\end{figure}

\subsection{Edge-Distinct Eulerian Trails}
This variant is straightforward. See also~\cite{TKDD2025}.
Compute the \emph{subdivided graph} of the input multigraph;
that is, the graph where every edge is subdivided into two edges by a new intermediate node.
The subdivided graph of a multigraph is a simple graph.
Apply \Cref{thm:main} to this graph.
It is easy to verify that the Eulerian trails in this graph are in one-to-one correspondence with the edge-distinct Eulerian trails in our input multigraph.
To ensure that we are able to output the correct encoding of these (i.e., the trie of edge-distinct Eulerian trails in our multigraph), observe that we can replace the newly added nodes from the subdivided graph and their adjacent edges with the corresponding edge from the input multigraph.

\begin{theorem}\label{the:edge-distinct}
    For any directed $m$-edge multigraph $G = (V, E)$ and any $z \in \N$, there is an algorithm that computes a compressed representation of $z$ edge-distinct Eulerian trails in $G$ (or all $\totalETs$ of them if $\totalETs\leq z$) in $\cO(m + \min(z, \totalETs))$ time.
\end{theorem}

\subsection{Node-Distinct Eulerian Trails}
For this variant, we need to adapt the compression rules to deal with the challenge illustrated in \Cref{fig:edge-v-node-distinct}.
These adaptations are natural and directly reflect the differences between the edge-distinct and node-distinct definitions.
Our overall algorithm will remain the same except for applying the new compression rules.
Therefore, we will only reprove the points of the argument affected by the changed compression rules. 

We fix a directed $m$-edge multigraph $G = (V, E)$ and a $z \in \N$ as our input.
For the algorithm solving this variant to be efficient, we use the following natural representation of the multigraph.
Instead of storing multiple identical edges between nodes, we will instead use a single edge along with its multiplicity as a natural number.
Whenever a single edge is deleted from $G$, the multiplicity of the corresponding multi-edge is reduced by one, and a multi-edge is completely removed if its multiplicity reaches zero.
This representation can be computed in $\cO(m)$ time from $G$ given as an edge list.
Importantly, we only apply this compaction once at the beginning of the algorithm.
During the execution of the algorithm, new parallel multi-edges might be introduced, which we will keep separate.

The definitions of branchings and crossings remain the same and have the same properties as in the simple graph variant.
When using one of the multi-edges $u \rightarrow v$ of the graph, if $u$ still has other in-edges to be used, then there must be a way to reach $u$ from $v$ after using this edge. Therefore, we make the following observation. 

\begin{observation} \label{obs:multicrossing-one}
    A crossing must have multiplicity exactly one.
\end{observation}

With this in hand, we proceed to the core lemma about when the state tree is branching.
It is the analogue of \cref{lem:crossings}.

\begin{lemma} \label{lem:multicrossings}
    A state $s$ of $T^u$ is branching if and only if the last node of $P_s$ has at least two outgoing multi-edges \ie{to distinct nodes} that are not crossings in $G'_s$.
\end{lemma}

\begin{proof}
The proof proceeds almost the same as for \cref{lem:crossings}.
Note that there is a subtlety about the degrees of the graph.
Concretely, in the multigraph setting, there may be two outgoing multi-edges, each of which represents multiple edges in the original graph.
In the argument for case $(3)$, the Euler criterion requires that we take the sum of multiplicities of incident edges correctly \eg{counting a multi-edge with multiplicity 2 twice}. 
If we do this correctly, the argument goes through the same way.
\end{proof}

Similarly, \cref{lem:incremental-scc} can be adapted as follows.

\begin{corollary}
    Given a multigraph $H$ and an Eulerian trail $R = e_1, \ldots, e_k$ in $H$, there is an algorithm running in $\cO(\setsize{R})$ time that determines which of the corresponding multi-edges to $e_i$ in the graph are crossings in $H \setminus \{e_1, \ldots, e_{i - 1}\}$.
\end{corollary}

We may now state the new compression rules adapted to multigraphs.

\begin{mcompression} \label{rule:m1}
    For any state $s$, if in the remaining graph $G_s'$, there is a node $v$ with a single outgoing multi-edge (of any multiplicity), contract that node and its only out-neighbor.
\end{mcompression}

\begin{mcompression} \label{rule:m2}
    Let $s$ be any state. If there is a node $v$ in $G_s'$ with only one self-loop, exactly one other outgoing multi-edge, and exactly one other incoming multi-edge, both of multiplicity one, remove the self-loop.
    If $v$ is the last node on $P_s$, we remove the self-loop only if $v$ has no other incoming multi-edge.
\end{mcompression}

The intuition for both rules is the same as in the simple graph version.
\cref{rule:m1} makes sure that trivially forced transitions are compressed, and \cref{rule:m2} removes loops if their starting point is never revisited.
Importantly, they yield the same guarantee that the state tree will be highly branching.

\begin{lemma}[Multigraph version of \Cref{lem:high-branching}] \label{lemma:multigraph-branching}
    Assume that in the state tree, \Cref{rule:m1,rule:m2} are applied exhaustively after each transition. Then there are never two forced state transitions in a row (i.e., there is no path of length two of forced transitions).
\end{lemma}
\begin{proof}
  The proof proceeds analogously to the simple graph version.
  The main reason is that \Cref{lem:multicrossings} tells us that, even in multigraphs, crossings will be edges of multiplicity one.
  Assume that $s_1, s_2, s_3$ are states in the compressed state tree such that the transitions $s_1 \to s_2$ and $s_2 \to s_3$ are forced. 
  Let $v_1, v_2, v_3$ be the last nodes on $P_{s_1},P_{s_2},P_{s_3}$, respectively.
  Our analysis proceeds in $G_{s_1}$.
  By \Cref{rule:m1}, $v_1$ and $v_2$ must have at least two outgoing multi-edges.
  As the transition is forced, for each node, one of them must have been a crossing (and thus have multiplicity one).
  The other edge can have arbitrary multiplicity.
  These two non-crossing edges cannot have been self-loops by \Cref{rule:m2}.
  In particular, $v_1$ and $v_2$ are distinct nodes, and the multi-edge $v_1 \to v_2$ is a non-crossing.
  Thus, $v_1$ and $v_2$ are in the same SCC.
  But then, this SCC has two outgoing crossings, a contradiction.
\end{proof}

After contraction, the graph may then have parallel multi-edges.
As we stated above, they should not be combined into a single edge, as the two multi-edges represent different node sequences in the original graph and Eulerian trail.
At its core, this is where we show why this version of the algorithm works specifically for node-distinct sequences (as opposed to the edge-distinct version from the previous section).
This is justified by the following lemma, which implies that each node in the state tree represents node-distinct prefixes of Eulerian trails.

\begin{lemma} \label{lemma:multigraph-disjoint-branching}
    Given a state $s$ of $T^u$, let $e_1$ and $e_2$ be two multi-edges in $G_s$. 
    Then, after decompression, they represent different node sequences in $G$.
\end{lemma}

\begin{proof}
  It is clear this is true if the start or end points of $e_1$ and $e_2$ differ, 
  as this will be the start and end node of the sequence these edges decompress into.
  So assume $e_1$ and $e_2$ are parallel multi-edges in $G_s$.
  Thus, they must have been created by an application of \Cref{rule:m1} (the other rule cannot create parallel multi-edges).
  Then, the statement of the lemma can easily be argued inductively over the applications of the rule.
\end{proof}

For the running time, \cref{rule:m1} is applied at most once in every state, as after taking a transition, the out-degree changes for only one node and \cref{rule:m2} is applicable only at the endpoints of the taken multi-edge, as the out-degree for those changes (similarly to \cref{lem:fast-compression}).
Thus, the running time stays the same.

Finally, as we have verified that the core lemmas from the simple version translate, we obtain the following result.

\begin{theorem}
    The state tree compressed using \Cref{rule:m1,rule:m2} has size $\cO(\totalETs)$.
    In particular, any state with $z$ leaf-descendants has $\cO(m+z)$ ancestors and descendants.
\end{theorem}

We have to be a bit more careful about analyzing the running time of this algorithm.
In particular, it could now be that \Cref{rule:m1} is applied to a node with only one out-neighbor but multiple different in-neighbors.
We can directly obtain an algorithm with running time $\cO((m + z) \log n)$, so near-optimal, by using a union-find data structure with undo support for the contraction operation \eg{\cite{DBLP:journals/siamcomp/WestbrookT89}}.

A more careful analysis shows that in fact even the naive implementation of the contraction operation (manually changing the pointers in the edge list) is sufficient to obtain the $\cO(m+z)$ running time.
We assume that during contraction, we re-point the pointers on the smaller side.
The core lemma for this improved analysis is the following.

\begin{lemma}
\label{lem:contraction-amortization}
Consider all states $s$ such that the transition to $s$ causes \Cref{rule:m1} to be applied at some node $v$.
If this caused $\ell$ pointers to be changed, then we can find $\Omega(\ell)$ different descendant states of $s$ such that these are injective over all $s$ \ie{the same state is not chosen for different $s$}.
\end{lemma}

Intuitively, one can think of the lemma giving us a way to charge the cost of the contraction operation at $v$ to states below $s$.
Since we charge at most once to each state, we get that the cost can be bounded by the number of states in our tree.
Note that this assignment is only an analysis tool, it is not calculated by the algorithm.

\begin{proof}
Assume $v$ has $k$ in-neighbors $u_1, \dots, u_k$ and notice that $k \ge \ell$.
For each in-neighbor $u_i$ of $v$ pick a descendant of $s$ that corresponds to using the multi-edge that originally was $u_i \to v$ before contractions happened.
Note that many different contractions could have happened, but since we do not merge multi-edges, the identity of this multi-edge persists.
Pick such a descendant as far up in the state tree as possible and only pick one that has not been picked before.
We argue below why there are always enough candidates.

Let $w$ be the node $v$ was contracted into.
First, observe that if for some $u_i$ there is a cycle returning to $w$ through $u_i$ that does not cause a contraction at $w$, then the desired assignment is trivially possible. 
The main challenge is that $w$ may be itself contracted away, meaning we might have to find a second assignment for the multi-edge from each $u_i$. 
We consider two cases: (1) $w$ has at least three out-neighbors, and (2) $w$ has at most two out-neighbors.

Case (1) is straightforward. Here, for each $u_i$ there must be a cycle through $w$ and $u_i$ (that is using the multi-edge we want to base our charging on) that only uses one of the out-multi-edges of $w$.
Therefore, there is a state that uses it before $w$ itself can be contracted away (since we never perform any of our compressions at a node with at least three out-neighbors).
We can simply use one of these states for our assignment.
Any potential future contraction of $w$ will be assigned to states below those, so that we achieve injectivity.

Case (2) is slightly more intricate.
If $w$ has out-degree 
at most two, then we could apply compression rules to it before any cycle including the edges from the neighbors $u_i$ returns to $w$, including possibly contracting $w$, leading us to need to find more states to charge the images of the previous multi-edges to.
Fortunately, using combinatorial insights, we can show that there will be enough possible states we can pick to still achieve injectivity.

We may assume that the prefix $\prefix{s}$ of the Eulerian trail ends in the same SCC as $v$ (because only there the degrees can change).
Since the graph is Eulerian, there are at least $k$ cycles including $v$, one for each in-neighbor.
As they can be traversed in any order, for each $u_i$, there are at least $k!$ descendant nodes corresponding to returning to the node that $v$ has been contracted into from $u_i$.
In particular, this means that for any $u_i$ we have $k!$ choices for our assignment (which is significantly more than the $\ell$ we need).
Now, since during contractions, we perform pointer changes on the smaller side, we may assume that if we now have performed $\ell'$ additional pointer changes on the in-side, the number of new in-neighbors of the node we contracted in is at least $\ell'$.
Thus, the new number of possible states we can assign our contractions to is $(k + \ell')!$.
Since $k \ge \ell$, we see $(k + \ell')!$ dominates $\ell + \ell'$, thus enabling us to still perform an injective assignment for the pointer changes on the in-side.
On the out-side we perform at most two pointer changes by the assumption of case (2). These, we will simply ignore. (This is why the lemma statement claims $\Omega(\ell)$ instead of $\ell$.)
This argument can be inductively repeated if more contractions occur at the node that $v$ belongs to before we reach the assigned states.
\end{proof}

With this observation in hand, we are able to argue that the increased contraction costs do not blow up the running time.
Recall that we explore the state tree in a DFS fashion: we have at most $m$ nodes along a single path in the state tree which are not fully explored.
For all states that cause a contraction where we have explored all descendants, we use \Cref{lem:contraction-amortization} to budget the contraction cost.
Since the assignment is injective and we construct $\cO(z)$ states, the total cost of the contractions budgeted for this is also $\cO(z)$.

For contractions in states where we have not explored all descendants, observe that since they lie on a single path, each node is contracted at most once.
By the Eulerian condition, if we have to perform $\ell$ pointer updates, then we also contract a multi-edge of multiplicity at least $\ell$.
Thus in total, the cost is bounded by the sum of all multiplicities, which is $m$.
Together, this gives us the desired $\cO(m + z)$ running time.

\begin{theorem}\label{the:node-distinct}
    For any directed $m$-edge multigraph $G = (V, E)$ and any $z \in \N$, there is an algorithm that computes a compressed representation of $z$ node-distinct Eulerian trails in $G$ (or all $\totalETs$ of them if $\totalETs\leq z$) in $\cO(m + \min(z, \totalETs))$ time.
\end{theorem}

\section{Final Remarks}\label{sec:final}
We have developed a remarkably simple and optimal algorithm for enumerating Eulerian trails in directed graphs (and multigraphs), essentially completing the complexity landscape of this fundamental topic. As a consequence, our result improves on an implementation of the BEST theorem for counting Eulerian trails when $\totalETs=o(n^2)$, and, in addition, it unconditionally improves the combinatorial $\cO(m\cdot \totalETs)$-time counting algorithm of Conte et al.~\cite{TKDD2025} for the same task. 

Beyond its theoretical or practical impact, our paper serves as a useful educational resource for graph algorithms classes, demonstrating how combinatorial insight yields \emph{simple} and \emph{optimal} solutions.

\section*{Acknowledgments}
We would like to thank Jeroen op de Beek
for reading an earlier version of this manuscript
and providing useful feedback.

\FloatBarrier
\bibliographystyle{alphaurl}
\bibliography{references}

\end{document}